\documentclass[10pt,a4paper,onecolumn,oneside]{article}

\usepackage[latin1]{inputenc}
\usepackage[x11names,dvipsnames]{xcolor}

\usepackage[affil-it]{authblk}

\usepackage{amsmath,amssymb,amsthm,stmaryrd,latexsym}
\usepackage[mathscr]{euscript}

\usepackage{booktabs}
\usepackage{hyperref}

\newcommand\anuncio[1]{!#1}

\newcommand\LanBox{{{\mathcal L}_\nc}}
\newcommand\LanFor{{{\mathcal L}_{\nc\ModForAbt}}}
\newcommand\AMA{{\mathbb A}}
\newcommand\AMD{{\mathbb D}}
\newcommand\LanForAct{{{\mathcal L}_{\nc\ModForAbt\AMD}}}

\newcommand\LogFor{{{\rm Log}_{\nc\ModForAbt}}}
\newcommand\LogForD{{d\text{-}{\rm Log}_{\nc\ModForAbt}}}
\newcommand\serial{\sf Ser}
\newcommand\model{\mathcal M}
\newcommand\modelsu{\models}
\newcommand\modelsd{\models_d}
\def\lb{\left\llbracket}
\def\rb{\right\rrbracket}
\newtheorem{definition}{Definition}
\newtheorem{proposition}{Proposition}
\newtheorem{example}{Example}
\newtheorem{theorem}{Theorem}
\newtheorem{lemma}{Lemma}
\newtheorem{corollary}{Corollary}
\newcommand\peq\preccurlyeq

\newcommand\nc{\Box}
\newcommand\ps{\Diamond}

\newcommand\ncforabt[2]{[{\ModForAbt #1 }]#2}
\newcommand\psforabt[2]{\langle {\ModForAbt #1 }\rangle #2}
\newcommand\ncfor[2]{[{\ModFor #1}]#2}
\newcommand\psfor[2]{\langle {\ModFor #1}\rangle #2}

\newcommand{\set}[1]{\left\{ #1 \right\}}
\newcommand{\tupla}[1]{\langle #1 \rangle}
\newcommand{\ts}[1]{\lb #1\rb}

\newcommand{\at}{\ensuremath{\mathit{At}}}
\newcommand{\D}{\ensuremath{\mathcal{C}}}
\newcommand{\su}{\ensuremath{u}}

\newcommand{\ModForAbt}{\boldsymbol{\ddagger}}
\newcommand{\ModFor}{\boldsymbol{\dagger}}
\newcommand{\limp}{\rightarrow}
\newcommand{\ldimp}{\leftrightarrow}

\newcommand{\U}{\mathsf{U}}
\newcommand{\E}{\mathsf{E}}
\newcommand{\R}{\mathsf{R}}
\newcommand{\e}{\mathsf{e}}
\newcommand{\f}{\mathsf{f}}

\newcommand{\pre}{\operatorname{Pre}}
\newcommand{\post}{\operatorname{Post}}

\newcommand\ncame[1]{[\U,\e]#1}
\newcommand\ncamf[1]{[\U,\f]#1}
\newcommand\ncam[3]{[#1,#2]#3}

\newcommand\ncfaln[3]{[#1, #2] #3}
\newcommand\ncfalt[3]{[\underline{#1}, #2] #3}
\newcommand\ncfalb[3]{[#1, \underline{#2}] #3}
\newcommand\ncfaldd[1]{[\_\_] #1}

\newcommand{\Udd}{\U_{\set{D_{1}, D_{2}}}}

\usepackage{tikz}
\usetikzlibrary{arrows,decorations,shapes,automata,positioning}

\tikzstyle{mundo}=[circle,draw,inner sep=0pt,minimum size=18pt]
\tikzstyle{texto}=[inner sep=0pt,minimum size=10pt]

\begin{document}

\title{Forgetting complex propositions}

\author[1]{David Fern\'andez--Duque}
\author[2]{\'Angel Nepomuceno--Fern\'andez}
\author[2]{Enrique Sarri\'on--Morrillo\textsuperscript{*}}
\author[2]{Fernando Soler--Toscano}
\author[2]{Fernando R. Vel\'azquez--Quesada}

\affil[1]{\small Department of Mathematics, Instituto Tecnol\'ogico Aut\'onomo de M\'exico, {\tt david.fernandez@itam.mx}}
\affil[2]{\small Grupo de L\'ogica, Lenguaje e Informaci\'on, Universidad de Sevilla, {\tt \{nepomuce,esarrion,fsoler,FRVelazquezQuesada\}@us.es}}
\affil[*]{\small \emph{Corresponding author.}}

\renewcommand\Authands{ and }

\maketitle

\begin{abstract}
This paper uses possible-world semantics to model the changes that may occur in an agent's knowledge as she loses information. This builds on previous work in which the agent may forget the truth-value of an atomic proposition, to a more general case where she may forget the truth-value of a propositional formula. The generalization poses some challenges, since in order to forget whether a complex proposition $\pi$ is the case, the agent must also lose information about the propositional atoms that appear in it, and there is no unambiguous way to go about this.

We resolve this situation by considering expressions of the form $\ncforabt\pi\varphi$, which quantify over all possible (but `minimal') ways of forgetting whether $\pi$. Propositional atoms are modified non-deterministically, although uniformly, in all possible worlds. We then represent this within action model logic in order to give a sound and complete axiomatization for a logic with knowledge and forgetting. Finally, some variants are discussed, such as when an agent {\em forgets} $\pi$ (rather than {\em forgets whether} $\pi$) and when the modification of atomic facts is done non-uniformly throughout the model.\\

\end{abstract}

\section{Introduction}

Epistemic notions such as knowledge and belief are subject to the effect of different epistemic actions, many of which have been studied in the literature. Just as beliefs can be affected by expansion \cite{Rott1989ctc}, contraction \cite{AlchouronGardenforMakinson1985,Fuhrmann1991}, revision \cite {Rott1989ctc,Boutilier1996irmccb,LeitgebSegerberg2007,vanBenthem2007dlbr,Rott2007,BaltagSmets2008tlg}, merging \cite {KoniecznyPinoPerez2011} and diverse forms of inference \cite{Velazquez2014delieb,NepomucenoEtAl2013abddel} among others, knowledge can be affected by deductive inference \cite{Velazquez2009InfUp,Velazquez2013eximNeigh}, public \cite{Plaza1989,GerbrandyGroeneveld1997} and other forms of announcements \cite{BaltagMossSolecki1999}. 

One action that has not received much attention is that of
\emph{forgetting} and its effect on an agent's \emph{knowledge}. One
of the reasons for this is its similarities with belief contraction,
an action that, when represented semantically, typically relies on
Lewis' system of spheres for conditionals \cite{Lewis1973}. This
system of spheres uses an ordering among theories (the theories'
`plausibility ordering') and thus provides a guideline for defining
the new beliefs an agent will have when one of the current ones is
discarded \cite{Grove1988}. This is adequate for belief contraction,
as a plausibility ordering is natural when defining beliefs: the
collection of epistemically possible situations can be understood as
having an order which not only defines this epistemic notion (as what
is true in the most plausible situations) but also establishes a
ranking among what is not believed but still has not been
discarded. However, such an ordering is not natural when dealing with
knowledge: there does not seem to be an ordering among the
epistemically possible situations that are known to \emph{not} be the
case and hence have been discarded.  

On the other hand, in the \emph{knowledge representation} area there are approaches for forgetting a finite set of atoms. In such proposals knowledge is represented as a finite set of formulas (the \emph{knowledge set}), and the typical definition of forgetting uses some notion of similarity between models: a knowledge set is the result of forgetting the atoms in $\at'$ if and only if every model of the resulting knowledge set is equivalent to a model of the original knowledge set when the atoms in $\at'$ are disregarded \cite{Lin94a}. In modal contexts as in this paper, the used equivalence notion is that of \emph{bisimulation}, which gives raise to systems \cite{ZhangZhou2008,DBLP:journals/ai/ZhangZ09} similar to those that contain modalities for bisimulation quantification \cite{french:2006}.

{\smallskip}

This work presents a logical treatment under possible worlds semantics
of an action that represents the forgetting of propositional formulas,
without relying on an ordering among theories or epistemic
possibilities and without using any notion of similarity notion between models. It can be seen as an extension of
\cite{vanDitmarschEtAl2009}, which deals only with forgetting the truth-value of atomic propositions. Several ways of forgetting a given formula
  are possible. We focus on two of them and give some hints about
  variants with different properties. The key intuition guiding our
  definitions is that an agent forgetting $\pi$ will lose her
  (possible) previous knowledge of $\pi$. But if she previously knew
  $\lnot\pi$, there are two possibilities after forgetting $\pi$: her
  knowledge of $\lnot\pi$ may fail or not. We call the first option {\em forgetting whether $\pi$} and the second {\em forgetting $\pi$.} While we will focus more on the first, we will also discuss the second possibility.

\paragraph{Layout of the paper} Section \ref{SecBasic} recalls some basic notions from propositional and epistemic logic which will be used throughout the text. Section \ref{SecUnif} presents the notion of {\em uniform forgetting whether} which, being the main focus of the article, is discussed in some detail in Section \ref{SecEffect}. Section \ref{SecFor} introduces a simpler action of `forgetting' (where a propositional formula is considered to be possibly false, but not necessarily possibly true) and compares it to the action of forgetting whether. Section \ref{SecAx} presents our main result, which is a sound and complete axiomatization for our logic of knowledge and forgetting. Finally, Section \ref{SecAlt} discusses some alternate ways of modelling the action of forgetting.

\section{Basic definitions}\label{SecBasic}

Our formalism for reasoning about forgetting will be based on {\em epistemic,} or more generally {\em modal,} logic. Throughout this text, {\at} denotes a designated countable non-empty set of {\em atoms} or {\em propositional variables.} Let us begin by reviewing the basic language of propositional modal logic.

\begin{definition}
  The grammar of $\LanBox$ is given by
  \begin{center}
    $\varphi::=\top\mid
    p\mid\neg\varphi\mid\left(\varphi\wedge\psi\right)\mid\nc \varphi$
  \end{center}
  where $p \in \at$. Formulas of the form $\nc\varphi$ are read as ``the agent knows that $\varphi$ is the case''. The symbols $\bot$, $\lor$, $\limp$, $\ldimp$ and $\ps$ are defined as usual.
\end{definition}

Modal logics are typically interpreted via their {\em Kripke} or {\em possible worlds} semantics, as described below:

\begin{definition}\label{def:model}
  A {\em Kripke frame} is a tuple $\mathcal F=\langle
  W,R\rangle$ where $W$ is a non-empty set and $R\subseteq
  W\times W$ a binary relation; no assumptions are made a priori about $R$. A {\em model} $\model = \tupla{\mathcal{F}, V}$ is a frame $\mathcal{F}$ equipped with a
  valuation $V:\at\rightarrow\mathcal{P}\left(W\right)$. A {\em pointed model} is a pair
  $(\model,w)$ with $\model$ a model and $w$ an
  element of its domain.
\end{definition}

\begin{definition}\label{DefSatBasic}
  Let $\model = \langle W,R,V\rangle$ be a model. The {\em satisfaction relation} $\models$ between pointed models and formulas is defined as follows:
  \begin{flushleft}
    \begin{tabular}{l@{\quad\text{iff}\quad}l}
      $\model,w\models p$                  & $w\in V(p)$; \\
      $\model,w\models \lnot \varphi$      & $\model,w\not\models \varphi$; \\
      $\model,w\models \varphi \land \psi$ & $\model,w\models \varphi$ {\;and\;} $\model,w\models \psi$; \\
      $\model,w\models \nc \varphi$        & for all $v\in W$, $wRv$ implies $\model,v\models \varphi$.
    \end{tabular}
  \end{flushleft}
Given a model $\model$, define the function $\ts{\cdot}^{\model}: \LanBox \rightarrow\mathcal{P}\left(W\right)$ as $w \in \ts{\varphi}^{\model}$ if and only if $\model,w\models \varphi$. The notation $\ts{\varphi}^{\model}$ will be abbreviated as $\ts{\varphi}$ when this does not lead to confusion.

As usual, $\model\models\varphi$ states that $\ts{\varphi}^\model=W$, and if $\sf X$ is a class of models, ${\sf X}\models \varphi$ states that  $\model\models\varphi$ for all $\model \in \sf X$. The formula $\varphi$ is {\em valid} when $\model\models\varphi$ for {\em every} model $\model$, a case denoted by $\models\varphi$.
\end{definition}

When modelling knowledge, the class of models in which the relation is an equivalence relation, $\sf S5$, is of particular interest. However, this paper will keep a more general discussion, only restricting its attention to models with particular properties when explicitly stated.

{\medskip}

In order to formalize our notion of forgetting, it will be convenient to represent propositional formulas in conjunctive normal form using sets of clauses. Recall that a \emph{literal} $\ell$ is an atom or its negation, and a {\em clause} $D$ is a finite (possibly empty) set of literals interpreted disjunctively, so that $D$ represents the formula $\bigvee D$\footnote{As usual, $\bigvee \varnothing := \bot$ and $\bigwedge \varnothing :=\top$.}. A clause $D$ is said to be a {\em consequence} of a propositional formula $\pi$ when $\models \pi \to\bigvee D$. 

A propositional formula is in {\em conjunctive normal form} when it is given as a finite (possibly empty) set of clauses, interpreted conjunctively. More precisely, the set of clauses $\D$ is interpreted as the formula $\widehat{\D}$ defined as
\[\widehat{\D} := \bigwedge_{D\in \D}\bigvee D.\]
Clearly, a given propositional formula may have many equivalent conjunctive normal forms, but
we wish to pick one canonically. In order to do so, first discard all \emph{tautological} clauses, i.e. those clauses $D$ in which there is an atom $p$ such that $\set{p,\neg p} \subseteq D$. A clause $D \neq \varnothing$ which is non-tautological is called
\emph{contingent}. Then, within each clause, `unnecessary' literals are removed: a clause $D$ is said to be a
{\em minimal consequence} of a propositional formula $\pi$ if and only if $\models \pi \limp \bigvee D$ and there is no $D'\subsetneq D$ such that
$\models \pi \limp \bigvee D'$. With this in mind, here is a formal definition:

\begin{definition}
  Let $\pi$ be a formula of propositional logic. Define the {\em clausal form} $\D(\pi)$ to be the set of all clauses that are minimal non-tautological consequences of $\pi$. Figure \ref{fig:exclausalforms} shows some examples.
\end{definition}

\begin{figure}
\[\begin{array}{cc@{\quad}|@{\quad}cc}
\toprule
\pi&\D(\pi)&\pi&\D(\pi)\\
\midrule
  p\land q&\set{\set{p},\set{q}}&
 \lnot(p\land q)&\set{\set{\lnot p,\lnot q}}\\
  p\lor q & \set{\set{p,q}}&
  \lnot(p\lor q) & \set{\set{\lnot p},\set{\lnot q}}\\
  p\limp q &\set{\set{\lnot p,q}} & \lnot (p\limp q)
  &\set{\set{p},\set{\lnot q}}\\
  p\ldimp q &\set{\set{\lnot p,q},\set{p,\lnot q}} & \lnot (p\ldimp q)
  &\set{\set{p,q},\set{\lnot p, \lnot q}}\\
  \bottomrule
  \end{array}\]
\caption{Some clausal forms that will be used in the text.}\label{fig:exclausalforms}
\end{figure}

The normal form $\D(\pi)$ is actually the set of prime implicates of $\pi$ (cf. \cite{Quine1952,DBLP:journals/jar/RameshBM97}), and there are several algorithms for calculating it (e.g., \cite{Quine1952,DBLP:conf/aaai/Kleer92,DBLP:journals/jsc/KeanT90,DBLP:journals/amai/Rymon94,DBLP:journals/jar/RameshBM97}; see \cite{DBLP:journals/logcom/Bittencourt08} for more). This concept has been already used for epistemic concerns, mainly on proposals following the {\it AGM} approach for belief revision \cite{AlchouronGardenforMakinson1985} in which the agent's beliefs are represented syntactically (e.g., \cite{DBLP:conf/ausai/Pagnucco06,ZhuangEtAl2007}). Here it will be used to simplify the model operation defined of the next section for representing the action of forgetting $\pi$.

{\medskip}

The next lemma is then straightforward.

\begin{lemma}\label{lem:clauses}
  For any propositional formula $\pi$, the set $\D(\pi)$ is finite, its elements are finite, and it satisfies $\models \pi \ldimp \widehat{\D}(\pi)$. Moreover, $\pi_1 \equiv \pi_2$ implies $\D (\pi_1) = \D (\pi_2)$, and $\D(\top) = \varnothing$ while $\D(\bot) = \set{\varnothing}$.
\end{lemma}

\section{Uniform forgetting}\label{SecUnif}

In order to reason about forgetting whether, the basic modal language will be extended with a new modality. 

\begin{definition}
  The language $\LanFor$ extends $\LanBox$ with expressions of the form $\ncforabt\pi\varphi$ with $\pi$ a propositional formula, read as ``after the agent forgets whether $\pi$, $\varphi$ is the case''. The expression $\psforabt{\pi}{\varphi}$ is defined in the standard way as  $\neg\ncforabt\pi{\neg\varphi}$.
\end{definition}

It is worthwhile to emphasise that, as discussed before, this paper understands ``forgetting whether $\pi$'' simply as ``forgetting $\pi$'s truth-value''; as such, the act of forgetting studied here does not involve other related actions (as, e.g., {\em becoming unaware of} atoms/formulas \cite{vanBenthemV10,vanDitmarschEtAl2013}).

{\medskip}

Observe how, since $\nc\varphi$ is the case when $\varphi$ holds in all of the agent's epistemic alternatives, in order for her to forget (i.e., to not know anymore) that a given propositional formula $\pi$ is the case, she needs to consider as possible at least one situation where $\pi$ fails. The first step is, then, to decide how to make a given $\pi$ fail in a given world $w$. Suppose that $\pi$ is not a tautology (such case will be discussed later). When $\pi$ is rewritten in its clausal form $\D(\pi) = \set{D_1,\ldots,D_n}$, it is clear than in order to make $\pi$ false at $w$, at least one clause in $\D(\pi)$ should be false in such a world; this would, in principle, give us a total of $2^n-1$ different forms of falsifying $\pi$. However, falsifying an arbitrary non-empty set of such clauses would be problematic, both because of the combinatorial explosion and because the negations of different clauses might be mutually inconsistent. A better alternative is to follow a \emph{minimal} change approach, where $\pi$ will be falsified by making only one of its clauses $D_i$ false.

Now, given the clause that will be falsified, we need to decide not only how many worlds need to be introduced as part of the agent's epistemic possibilities, but also which truth-value will be assigned, in such new worlds, to the atoms that do not appear in the given clause. Again, the minimal change approach suggests that the least intrusive way to change the agent's knowledge is to make a copy of the current epistemic possibilities and then falsify the given clause in each one of them.\footnote{Of course, such operation is not minimal with respect to the number of worlds that will be added; it is minimal with respect to the changes in the agent's knowledge.} In the resulting model, the original formula $\pi$ has been \emph{uniformly} falsified because the same clause $D_i \in \D(\pi)$ has been falsified across all the worlds in the new copy of the set of epistemic possibilities. 

{\smallskip}

The formalisation of this idea will be used to provide the semantic interpretation of formulas expressing the effect of the slightly different `forgetting whether $\pi$', $\ncforabt\pi\varphi$. In order for an agent to forget the truth-value of a given $\pi$, she needs to consider not only a possibility that falsifies $\pi$ (by falsifying one of the clauses of $\pi$'s clausal form) but also a possibility that falsifies $\lnot \pi$ (by falsifying one of the clauses of $\lnot\pi$'s clausal form). Thus, a model operation representing this action takes two clauses and creates two copies of the current set of epistemic possibilities, with each copy falsifying one clause. The operation defined below is a generalisation that receives an epistemic model and a finite set of clauses $\D$, returning a model with a copy of the current set of epistemic possibilities falsifying each clause in $\D$. 

\begin{definition}\label{def:multiclause}
  Let $\model = \langle W, R, V \rangle$ be a model
  and $\D=\{D_i : i\in I\}$ a finite (possibly empty) set of 
  non-tautological clauses, where without loss of generality $0\not\in I$. The new model $\model^{\D}_\su = \tupla{
    W^{\D}_\su, R^{\D}_\su, 
    V^{\D}_\su}$ is defined as follows:
    \begin{enumerate}
  \item $W^{\D}_\su := W \times (\set{0} \cup I);$
   \item for all $w,v\in W$ and $i,j\in \{0\}\cup I$, $(w,i)R^{\D}_\su (v,j)$ if and only if $wRv$;
    \item for all $w\in W$, $(w,0) \in V^{\D}_\su(p)$ if and only if $w \in V(p)$; and
    \item for all $w\in W$ and $i\in I$, $(w,i) \in V^{\D}_\su(p) $ if and only if one of the following holds:
   \begin{enumerate}
   
   \item $\neg p \in D_i$; or
   \item both $\set{p, \lnot p} \cap D_i = \varnothing$ and $w \in V(p)$.
   
   \end{enumerate}
\end{enumerate}
\end{definition}

Thus, $W^{\D}_\su$ has two types of worlds. Worlds of the form $(w, 0)$ preserve the original valuation: an atom $p$ is true on $(w,0)$ if and only if $p$ was already true on $w$. On the other hand, each world of the form $(w, i)$ with $i\in I$ falsifies all of the literals in $D_i$, leaving the remaining atoms as before. The relation in the new model simply follows the original relation, making a world $(v,j)$ accessible from a world $(w,i)$ when $v$ is accessible from $w$ in the original model. 

Note also that we are modelling forgetting within the context of epistemic logic, where knowledge is represented semantically. This leads to several key differences from syntactic approaches of knowledge representation. Most notably, an agent cannot distinguish between semantically equivalent formulas (so, if she knows $\pi$, she also knows all its semantic equivalents). By using both $\pi$'s and $\lnot \pi$'s minimal clausal forms, the \emph{forgetting whether} action treats semantically equivalent formulas in the same way (so, afterwards, the agent has forgotten not only $\pi$'s truth value, but also that of all $\pi$'s semantic equivalents). Approaches that distinguish semantically equivalent formulas are possible, but would require a different framework for modelling the agent's knowledge.

{\medskip}

With this model operation it is possible to define the semantic interpretation of formulas of the form $\ncforabt\pi\varphi$ which, it is recalled, are intuitively read as ``after the agent forgets the truth-value of $\pi$, $\varphi$
is the case''. 

\begin{definition} \label{def:forabt} 
  Let $\model = \langle W, R, V \rangle$ be a model and $w$ a world of
  $W$. We extend 
  Definition \ref{DefSatBasic} to $\LanFor$ by defining $\model, w
  \modelsu \ncforabt\pi\varphi$ if and only if, for all $D_1\in \D(\pi)$ and $D_2\in \D(\neg \pi),$ \[\model^{\set{D_1,D_2}}_\su, (w,0) \modelsu\varphi.
  \] 
  The set of formulas in $\LanFor$ valid under $\modelsu$ will be denoted $\LogFor$.
\end{definition}

Thus, $\ncforabt\pi\varphi$ states that $\varphi$ is the case after the
agent forgets the truth value of $\pi$, independently of the
choice of the clauses $D_1 \in \D(\pi)$ and $D_2 \in \D(\lnot\pi)$ that are
falsified in the added worlds.

\section{The effect of \emph{forgetting whether}}\label{SecEffect}

The model $\model^{\D}_\su$ is the result of the agent considering new possibilities in which each clause in $\D$ fails. This is achieved by keeping a copy of the original model (the $(w, 0)$-worlds, which preserve the original valuation) and adding, for each clause $D_{i}$, a copy of the original model (the $(w, i)$-worlds) in which $D_{i}$ is falsified by falsifying each of its literals in {\em all} of the worlds in the copy, keeping the remaining atoms as before. Thus, each clause $D_{i}$ is false at each world $(w, i)$, as the following lemma shows.

\begin{lemma}\label{lem:clauseFails}
  Let $\model = \langle W, R, V \rangle$ be a model
  and $\D =\set {D_i : i\in I}$ a non-empty finite set of
  non-tautological clauses (again $0 \notin I$). Then, for any $w \in W$ and any $i\in I$,  
  \[ \model^\D_\su, (w,i) \not\models \bigvee D_i .\]
\end{lemma}

  \begin{proof}
    If $D_i=\varnothing$ the result is trivial, as $\bigvee D_i =
    \bot$. Otherwise take $D_i = \set{l_1, \ldots, l_m}$ (i.e., $D_i$ is contingent). Then, 
    for any $l_k \in D_i$,  
    \begin{itemize}
      \item if $l_k$ is an atom $p$, then since $D_i$ is contingent, $\lnot p \not\in D_i$; thus, by definition, $(w,i) \not\in V^{\D}_\su(p)$ and hence $(w,i) \not\in \ts{l_k}^{\model^\D_\su}$.
      \item if $l_k$ is an atom's negation $\lnot p$, then $\lnot p \in D_i$; so, by definition, $(w,i) \in V^{\D}_\su(p)$ and thus $(w,i) \not\in \ts{l_k}^{\model^\D_\su}$.
    \end{itemize}
    Hence, every literal in $D_i$ fails at $(w,i)$ and therefore so the disjunction $\bigvee D_i$.
  \end{proof}

\begin{example}\label{exa:seriality}
  Consider the following pointed model $(\model, w_{0})$ (with each world $w$ containing $V(w)$ and the evaluation double-circled) in which the agent knows $p$ (i.e., $\model, w_{0} \modelsu \nc p$):
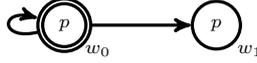
  \begin{center}
    \begin{tikzpicture}[->,>=stealth', very thick]
      \node [mundo, double] (w0) at (-1, 0) {\scriptsize $p$};
      \node [mundo] (w1) at ( 1, 0) {\scriptsize $p$};

      \node [texto, node distance=7pt, right=of w0.south] {\scriptsize $w_{0}$};
      \node [texto, node distance=7pt, right=of w1.south] {\scriptsize $w_{1}$};

      \path (w0) edge [loop left] (w1)
                 edge (w1);
    \end{tikzpicture}
  \end{center}
  Consider the action of forgetting whether $p$. Given that $\D(p)=
    \set{ \set{p} }$, there is only one clause to be chosen:
  $\set{p}$. Similarly, $\D(\neg p)=
    \set{ \set{\neg p} }$, so the only clauses the agent will consider when forgetting whether $p$ are $D_1=\{p\}$ and $D_2=\{\neg p\}$. The pointed model $(\model^{\set{\set{p}, \set{\lnot p}}}_\su, (w_{0}, 0))$ appears below, with
  the top row being the copy that results from making
  $\set{p}$ false and the bottom
  row being the copy that results from making $\set{\lnot p}$ false
  (thus making $p$ true in those worlds):
  
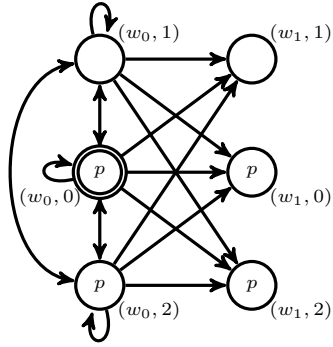
\begin{center}
    \begin{tikzpicture}[->,>=stealth', very thick]
      \node [mundo] (w10) at (-1, 1.5) {};
      \node [mundo] (w11) at ( 1, 1.5) {};
      \node [mundo, double] (w00) at (-1, 0) {\scriptsize $p$};
      \node [mundo] (w01) at ( 1, 0) {\scriptsize $p$};
      \node [mundo] (w20) at (-1,-1.5) {\scriptsize $p$};
      \node [mundo] (w21) at ( 1,-1.5) {\scriptsize $p$};

      \node [texto, node distance=6pt, right=of w10.north] {\scriptsize $(w_{0},1)$};
      \node [texto, node distance=6pt, right=of w11.north] {\scriptsize $(w_{1},1)$};
      \node [texto, node distance=6pt, left=of w00.south] {\scriptsize $(w_{0},0)$};
      \node [texto, node distance=6pt, right=of w01.south] {\scriptsize $(w_{1},0)$};
      \node [texto, node distance=6pt, right=of w20.south] {\scriptsize $(w_{0},2)$};
      \node [texto, node distance=6pt, right=of w21.south] {\scriptsize $(w_{1},2)$};

      \path (w10) edge [loop above] (w11)
                  edge (w11)
            (w00) edge [loop left] (w01)
                  edge (w01)
            (w20) edge [loop below] (w21)
                  edge (w21)
            (w10) edge [<->] (w00)
                  edge [<->, bend right = 80] (w20)
                  edge (w01)
                  edge (w21)
            (w00) edge (w11)
                  edge (w21)
            (w20) edge [<->] (w00)
                  edge (w01)
                  edge (w11);

    \end{tikzpicture}
  \end{center}
  As a result of the action, the agent considers possible worlds where
  $p$ holds as well as worlds where $p$ fails. Thus, $\model^{\set{\set{p}, \set{\lnot p}}}_\su, (w_{0}, 0) \modelsu \lnot \nc
  p \land \lnot \nc{\lnot p}$, and hence $\model, w_{0} \modelsu
  \nc p \land \ncforabt{p}{(\lnot \nc p \land \lnot \nc{\lnot p})}$.
\end{example}

In the previous example, note how, if $w_{1}$ were the evaluation point at the initial model (and hence $(w_{1},0)$ the evaluation point at the model after the operation), then the agent would know $p$ before the action (by vacuity), but she would still know $p$ afterwards (by vacuity too). The following proposition shows that this counterintuitive outcome of the forgetting whether action can only occur when the knowledge of the agent is inconsistent to begin with.

\begin{proposition}\label{PropConsK}
  Let $\pi$ be a propositional formula that is neither a tautology nor a contradiction (so $\D(\pi)$ and $\D(\lnot\pi)$ are both non-empty sets of contingent clauses). Then, 
  \[ \modelsu \ncforabt{\pi}{(\nc{\lnot\pi} \lor \nc{\pi})} \ldimp \nc{\bot} .\]
\end{proposition}

  \begin{proof}
    Let $(\model,w)$ be a pointed model with $\model = \langle W,R,V \rangle$.

    From right to left, suppose $\model,w \modelsu \nc{\bot}$. Then there is no $v$ such that $wRv$ and hence by $R^\D_\su$'s definition, and regardless of $\D$, there is no $(v, i)$ such that $(w, 0)R^\D_\su(v, i)$. Hence ${\cal M}^{\D}_\su, (w, 0) \modelsu \nc{\lnot \pi} \lor \nc\pi$ and therefore $\model,w \modelsu \ncforabt{\pi}{(\nc{\lnot \pi} \lor \nc \pi)}$.

    The left-to-right direction is proved by contrapositive. Suppose that $\model,w \modelsu \lnot\nc{\bot}$; then there is $v$ such that
    $wRv$. By Definition \ref{def:multiclause}, for any $D_1 \in \D(\pi)$ and $D_2 \in
    \D(\lnot\pi)$, $(w,0)R^{\set{D_1, D_2}}_{\su}(v,1)$ and $(w,0)R^{\set{D_1,
        D_2}}_{\su}(v,2)$. By the
    contingency of $D_1$ and $D_2$ and Lemma \ref{lem:clauseFails}, $\model^{\set{D_1,
        D_2}}_{\su}, (v,i) \not\modelsu \bigvee D_i$ for $i \in \set{1, 2}$ and hence
    both $\model^{\set{D_1, D_2}}_{\su}, (v,1) \not\modelsu \widehat{\D}(\pi)$ and
    $\model^{\set{D_1, D_2}}_{\su}, (v,2) \not\modelsu \widehat{\D}({\neg \pi})$. Then
    $\model^{\set{D_1, D_2}}_{\su}, (w,0) \modelsu \ps{\lnot \pi} \land
    \ps{\pi}$ and therefore, since neither $\D(\pi)$ nor
    $\D(\lnot\pi)$ is empty, $\model,w \modelsu
    \psforabt{\pi}{(\ps \pi \land \ps{\lnot \pi})}$, i.e., $\model,w \not\modelsu \ncforabt{\pi}{(\nc \pi \vee \nc{\lnot \pi})}$. 
  \end{proof} 

As a special case, if $\pi$ is an atom $p$, both $\D(p) = \set{ \set{p} }$ and $\D(\lnot p) = \set{ \set{\lnot p} }$ are non-empty and both contain only contingent clauses, so from the above proposition it follows that $ \ncforabt{p}{(\nc{p} \lor \nc{\lnot p})} \ldimp \nc{\bot} $ is valid. More interestingly, recall that an agent's knowledge is consistent at $w$ if and only if $w$ has at least one accessible world. In the class of models where this consistency property holds, called \emph{serial} and denoted by $\serial$, we obtain a stronger version of Proposition \ref{PropConsK}. 

\begin{corollary}
  For any non-tautological and non-contradictory propositional formula $\pi$,
  \[ \serial \modelsu \psforabt{\pi}{\top} \land \ncforabt{\pi}{(\lnot \nc \pi \land \lnot \nc{\lnot \pi})}. \]

\end{corollary}

Note that Proposition \ref{PropConsK} is restricted to formulas $\pi$ that are neither tautologies nor contradictions because otherwise the proof does not go through: in such cases either $\D(\pi)=\varnothing$ or else $\D(\lnot \pi)=\varnothing$, and hence there are no clauses for falsifying one of $\pi$ or $\lnot \pi$. As a consequence of this behaviour, both $\ncforabt{\top}\varphi$ and $ \ncforabt{\bot}\varphi$ are valid for any formula $\varphi$, and thus neither $\psforabt{\top}\top$ nor $\psforabt{\bot}\top$ is satisfiable.

{\medskip}

\begin{example}\label{exa:nobisimulation}
  Consider the following pointed model in which the agent knows $p \limp q$.
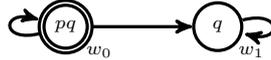
  \begin{center}
    \begin{tikzpicture}[->,>=stealth', very thick]
      \node [mundo, double] (w0) at (-1, 0) {\scriptsize $pq$};
      \node [mundo] (w1) at ( 1, 0) {\scriptsize $q$};

      \node [texto, node distance=7pt, right=of w0.south] {\scriptsize $w_{0}$};
      \node [texto, node distance=7pt, right=of w1.south] {\scriptsize $w_{1}$};

      \path (w0) edge [loop left] (w1)
                 edge (w1)
            (w1) edge [loop right] (w1);
    \end{tikzpicture}
  \end{center}
  Consider the act of forgetting whether $p \limp q$. Since $\D(p \limp q) = \set{\set{\lnot p, q}}$ and $\D(\lnot(p \limp q)) = \set{\set{p} \set{\lnot q}}$, there are two possible outcomes:
  \begin{center}
    \begin{tabular}{c@{\qquad\qquad}c}
      By using $\set{\lnot p, q}$ and $\set{p}$:
      &
      By using $\set{\lnot p, q}$ and $\set{\lnot q}$:
      \\
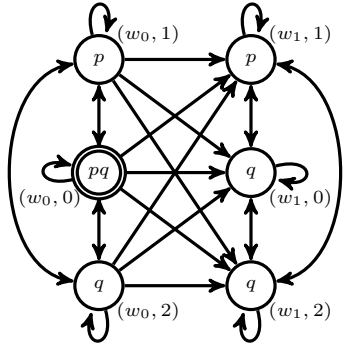
      \begin{tikzpicture}[->,>=stealth', very thick]
        \node [mundo] (w10) at (-1, 1.5) {\scriptsize $p$};
        \node [mundo] (w11) at ( 1, 1.5) {\scriptsize $p$};
        \node [mundo, double] (w00) at (-1, 0) {\scriptsize $pq$};
        \node [mundo] (w01) at ( 1, 0) {\scriptsize $q$};
        \node [mundo] (w20) at (-1,-1.5) {\scriptsize $q$};
        \node [mundo] (w21) at ( 1,-1.5) {\scriptsize $q$};

        \node [texto, node distance=6pt, right=of w10.north] {\scriptsize $(w_{0},1)$};
        \node [texto, node distance=6pt, right=of w11.north] {\scriptsize $(w_{1},1)$};
        \node [texto, node distance=6pt, left=of w00.south] {\scriptsize $(w_{0},0)$};
        \node [texto, node distance=6pt, right=of w01.south] {\scriptsize $(w_{1},0)$};
        \node [texto, node distance=6pt, right=of w20.south] {\scriptsize $(w_{0},2)$};
        \node [texto, node distance=6pt, right=of w21.south] {\scriptsize $(w_{1},2)$};

        \path (w10) edge [loop above] ()
                    edge (w11)
                    edge [<->] (w00)
                    edge [<->, bend right = 80] (w20)
                    edge (w01)
                    edge (w21)
              (w00) edge [loop left] ()
                    edge (w01)
                    edge (w11)
                    edge (w21)
              (w20) edge [loop below] ()
                    edge (w21)
                    edge [<->] (w00)
                    edge (w01)
                    edge (w11)
              (w11) edge [loop above] ()
                    edge [<->] (w01)
                    edge [<->, bend left = 80] (w21)
              (w01) edge [loop right] ()
              (w01) edge [<->] (w21)
              (w21) edge [loop below] ();
      \end{tikzpicture}
      &
      
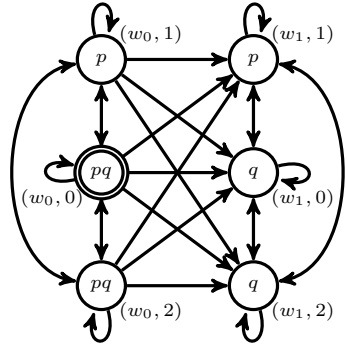
\begin{tikzpicture}[->,>=stealth', very thick]
        \node [mundo] (w10) at (-1, 1.5) {\scriptsize $p$};
        \node [mundo] (w11) at ( 1, 1.5) {\scriptsize $p$};
        \node [mundo, double] (w00) at (-1, 0) {\scriptsize $pq$};
        \node [mundo] (w01) at ( 1, 0) {\scriptsize $q$};
        \node [mundo] (w20) at (-1,-1.5) {\scriptsize $pq$};
        \node [mundo] (w21) at ( 1,-1.5) {\scriptsize $q$};

        \node [texto, node distance=6pt, right=of w10.north] {\scriptsize $(w_{0},1)$};
        \node [texto, node distance=6pt, right=of w11.north] {\scriptsize $(w_{1},1)$};
        \node [texto, node distance=6pt, left=of w00.south] {\scriptsize $(w_{0},0)$};
        \node [texto, node distance=6pt, right=of w01.south] {\scriptsize $(w_{1},0)$};
        \node [texto, node distance=6pt, right=of w20.south] {\scriptsize $(w_{0},2)$};
        \node [texto, node distance=6pt, right=of w21.south] {\scriptsize $(w_{1},2)$};

        \path (w10) edge [loop above] () 
                    edge (w11)
                    edge [<->] (w00)
                    edge [<->, bend right = 80] (w20)
                    edge (w01)
                    edge (w21)
              (w00) edge [loop left] ()
                    edge (w01)
                    edge (w11)
                    edge (w21)
              (w20) edge [loop below] ()
                    edge (w21)
                    edge [<->] (w00)
                    edge (w01)
                    edge (w11)
              (w11) edge [loop above] ()
                    edge [<->] (w01)
                    edge [<->, bend left = 80] (w21)
              (w01) edge [loop right] ()
              (w01) edge [<->] (w21)
              (w21) edge [loop below] ();
      \end{tikzpicture}
    \end{tabular}
  \end{center}
  Observe how $\lnot \nc{(p \limp q)}$ holds in the two pointed models, as in both 
  the agent considers possible a world where $p$ holds but $q$
  fails. Similarly, $\lnot \nc{\lnot (p \limp q)}$ holds in the two cases, as in both
  models the agent considers possible worlds where $p \limp q$
  holds. 

Still, the two pointed models do not represent the same state of knowledge or, to be precise, they are not bisimilar. In the model on the left, the agent considers possible a world where $\lnot p \land q$ holds and $p \land q$ is possible (the path $(w_0, 0) \to (w_0, 2) \to (w_0, 0) $), a possibility that does not exists in the model on the right (any transition from $(w_0, 0)$ to a $\lnot p \land q$-world forces a move to the right-hand side of the diagram, from which there are no arrows back to the left-side). Thus,
  \[
    \model^{\set{\set{\lnot p, q}, \set{p}}}_\su, (w_{0}, 0) \modelsu \ps{(\lnot p \land q \land \ps{(p \land q)})}
  \]
    but
  \[
    \model^{\set{\set{\lnot p, q}, \set{\lnot q}}}_\su, (w_{0}, 0) \not\modelsu \ps{(\lnot p \land q \land \ps{(p \land q)})}
  \]
  and hence
  \[
    \model, w_0 \not\modelsu \ncforabt{(p \limp q)}{\ps{(\lnot p \land q \land \ps{(p \land q)})}}.
  \]
  
\end{example}

\begin{example}\label{exa:forgetnotknown}
  Consider again the initial pointed model of Example
  \ref{exa:nobisimulation}. Observe how the agent knows neither $p
  \land q$ (she considers $w_1$ possible) nor $\lnot (p \land q)$ (she
  considers $w_0$ possible). Since $\D(p \land q) = \set{\set{p},
    \set{q}}$ and $\D(\lnot(p \land q)) = \set{\set{\lnot p, \lnot
      q}}$, there are two possible outcomes for an action of
  forgetting the truth-value of the already `unknown' $p \land q$: 
  \begin{center}
    \begin{tabular}{c@{\qquad\qquad}c}
      By using $\set{p}$ and $\set{\lnot p, \lnot q}$:
      &
      By using $\set{q}$ and $\set{\lnot p, \lnot q}$:
      \\
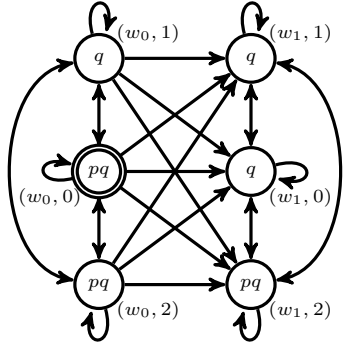
      \begin{tikzpicture}[->,>=stealth', very thick]
        \node [mundo] (w10) at (-1, 1.5) {\scriptsize $q$};
        \node [mundo] (w11) at ( 1, 1.5) {\scriptsize $q$};
        \node [mundo, double] (w00) at (-1, 0) {\scriptsize $pq$};
        \node [mundo] (w01) at ( 1, 0) {\scriptsize $q$};
        \node [mundo] (w20) at (-1,-1.5) {\scriptsize $pq$};
        \node [mundo] (w21) at ( 1,-1.5) {\scriptsize $pq$};

        \node [texto, node distance=6pt, right=of w10.north] {\scriptsize $(w_{0},1)$};
        \node [texto, node distance=6pt, right=of w11.north] {\scriptsize $(w_{1},1)$};
        \node [texto, node distance=6pt, left=of w00.south] {\scriptsize $(w_{0},0)$};
        \node [texto, node distance=6pt, right=of w01.south] {\scriptsize $(w_{1},0)$};
        \node [texto, node distance=6pt, right=of w20.south] {\scriptsize $(w_{0},2)$};
        \node [texto, node distance=6pt, right=of w21.south] {\scriptsize $(w_{1},2)$};

        \path (w10) edge [loop above] ()
                    edge (w11)
                    edge [<->] (w00)
                    edge [<->, bend right = 80] (w20)
                    edge (w01)
                    edge (w21)
              (w00) edge [loop left] ()
                    edge (w01)
                    edge (w11)
                    edge (w21)
              (w20) edge [loop below] ()
                    edge (w21)
                    edge [<->] (w00)
                    edge (w01)
                    edge (w11)
              (w11) edge [loop above] ()
                    edge [<->] (w01)
                    edge [<->, bend left = 80] (w21)
              (w01) edge [loop right] ()
              (w01) edge [<->] (w21)
              (w21) edge [loop below] ();
      \end{tikzpicture}
      &
      
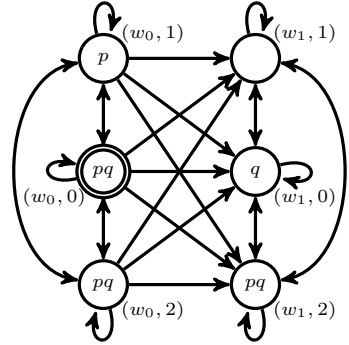
\begin{tikzpicture}[->,>=stealth', very thick]
        \node [mundo] (w10) at (-1, 1.5) {\scriptsize $p$};
        \node [mundo] (w11) at ( 1, 1.5) {};
        \node [mundo, double] (w00) at (-1, 0) {\scriptsize $pq$};
        \node [mundo] (w01) at ( 1, 0) {\scriptsize $q$};
        \node [mundo] (w20) at (-1,-1.5) {\scriptsize $pq$};
        \node [mundo] (w21) at ( 1,-1.5) {\scriptsize $pq$};

        \node [texto, node distance=6pt, right=of w10.north] {\scriptsize $(w_{0},1)$};
        \node [texto, node distance=6pt, right=of w11.north] {\scriptsize $(w_{1},1)$};
        \node [texto, node distance=6pt, left=of w00.south] {\scriptsize $(w_{0},0)$};
        \node [texto, node distance=6pt, right=of w01.south] {\scriptsize $(w_{1},0)$};
        \node [texto, node distance=6pt, right=of w20.south] {\scriptsize $(w_{0},2)$};
        \node [texto, node distance=6pt, right=of w21.south] {\scriptsize $(w_{1},2)$};

        \path (w10) edge [loop above] () 
                    edge (w11)
                    edge [<->] (w00)
                    edge [<->, bend right = 80] (w20)
                    edge (w01)
                    edge (w21)
              (w00) edge [loop left] ()
                    edge (w01)
                    edge (w11)
                    edge (w21)
              (w20) edge [loop below] ()
                    edge (w21)
                    edge [<->] (w00)
                    edge (w01)
                    edge (w11)
              (w11) edge [loop above] ()
                    edge [<->] (w01)
                    edge [<->, bend left = 80] (w21)
              (w01) edge [loop right] ()
              (w01) edge [<->] (w21)
              (w21) edge [loop below] ();
      \end{tikzpicture}
    \end{tabular}
  \end{center}

  In both resulting pointed models the agent still knows neither $p \land q$
  nor $\lnot (p \land q)$. However, in both cases the action has
  an effect on the agent's information: in the leftmost pointed model she
  considers possible a $\lnot p \land q$-world, $(w_0,1)$, from
  which there is an accessible $p \land q$-world, 
  $(w_0,0)$, something she did not consider possible before:
  \begin{center}
    \begin{small}
      \begin{tabular}{@{}c@{\;\;\;}c@{\;\;\;}c@{}}
        $\model^{\set{\set{p}, \set{\lnot p, \lnot q}}}_\su, (w_{0}, 0) \modelsu \ps{(\lnot p \land q \land \ps{(p \land q)})}$
        &
        but
        &  
        $\model, w_{0} \not\modelsu \ps{(\lnot p \land q \land \ps{(p \land q)})}.$
      \end{tabular}
    \end{small}
  \end{center}
  Moreover, in the rightmost pointed model she considers possible a $p \land
  \lnot q$-world, $(w_0,1)$, something she did not consider
  possible before:  
  \begin{center}
    \begin{small}
      \begin{tabular}{c@{\qquad}c@{\qquad}c}
        $\model^{\set{\set{q}, \set{\lnot p, \lnot q}}}_\su, (w_{0}, 0) \modelsu \ps{(p \land \lnot q)}$
        &
        yet 
        &
        $\model, w_{0} \not\modelsu \ps{(p \land \lnot q)}.$ 
      \end{tabular}
    \end{small}
  \end{center} 
  Thus, forgetting the truth-value of formulas whose truth-value is not known to begin with can affect the agent's information by giving her `new reasons' to not know the formula's truth-value.  
\end{example}

\section{A simpler `forgetting' action}\label{SecFor}

In the current setting it is straightforward to define a simpler action that, instead of forgetting $\pi$'s truth value, simply {\em forgets that} $\pi$ is the case. For this, it is enough to use the model operation of Definition \ref{def:multiclause} omitting the clause for $\D(\lnot\pi)$. Here are the formal definitions:

\begin{definition}
  The language $\mathcal{L}_{\Box\ModFor}$ extends $\mathcal{L}_{\Box}$ with operators of the form $\ncfor{\pi}{}$ for $\pi$ a propositional formula, thus allowing the construction of formulas of the form $\ncfor\pi\varphi$, read as ``after the agent {\em forgets} $\pi$, $\varphi$ is the case''. 
\end{definition}

\begin{definition}\label{def:for}
Let $\model = \langle W, R, V \rangle$ be a model, $w\in W$ and $\pi$ be a propositional formula. We extend Definition \ref{DefSatBasic} by setting
  \[
    \model, w \modelsu \ncfor\pi\varphi
    \quad\text{iff}\quad
    \forall D\in \D(\pi), \model^{\set{D}}_\su, (w,0) \modelsu\varphi.
  \]
\end{definition}

This shows how the forgetting whether $\pi$ action of before consists of simultaneously forgetting both $\pi$ and $\neg\pi$. The question naturally arises of whether the action of forgetting $\pi$'s truth-value could instead be defined as forgetting $\pi$ and {\em then} forgetting $\lnot
\pi$. Below it is shown that this is not the case.

\begin{proposition} The expressions $\ncforabt{\pi}{\varphi}$ and $\ncfor{\pi}{\ncfor{\lnot \pi}{\varphi}}$ are not equivalent, even over the class of $\sf S5$ models.
  \begin{proof}
  Consider the following pointed model $(\model,w_0)$ with both $p$
    and $q$ false at $w_0$:
    
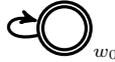
\begin{center}
    \begin{tikzpicture}[->,>=stealth', very thick]
      \node [mundo, double] (w0) at (0, 0) {\mbox{}};
      \node [texto, node distance=10pt, right=of w0.south]
      {\scriptsize $w_{0}$}; 
      \path (w0) edge [loop left] (w0);
    \end{tikzpicture}
  \end{center}
  Now, let $\pi$ be $\lnot p \lor \lnot q$. So, $\D(\pi) = \set{\set{\lnot
  p,\lnot q}}$ and $\D(\lnot \pi) = \set{\set{p},\set{q}}$. Then, by
  using first $\set{\lnot p,\lnot q}$ in $\D(\pi)$ and then $\set{q}$
  in $\D(\lnot\pi)$, we build 
  $(\model^{\set{\set{\lnot p,\lnot q}}}_\su)^{\set{\set{q}}}_\su$ 
  in the following way: 
  \begin{center}
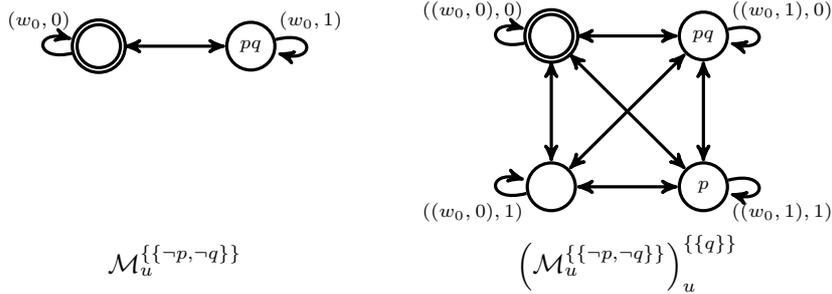
    \begin{tabular}{c@{\qquad\quad}c}
    \begin{tikzpicture}[->,>=stealth', very thick]
      \node [mundo, double] (w00) at (0, 0) {\mbox{}};
      \node [mundo] (w01) at (2, 0) {\scriptsize $pq$};
  
      \node [texto] at (1,-2.2) {\mbox{}};

      \node [texto, node distance=10pt, left=of w00.north]
      {\scriptsize $(w_{0},0)$}; 
      
      \node [texto, node distance=10pt, right=of w01.north]
      {\scriptsize $(w_{0},1)$}; 
      \path (w00) edge [loop left] (w00)
                  edge [<->] (w01)
            (w01) edge [loop right] (w01);
    \end{tikzpicture}
    &
    \begin{tikzpicture}[->,>=stealth', very thick]
      \node [mundo, double] (w000) at (0, 0) {\mbox{}};
      \node [mundo] (w010) at (2, 0) {\scriptsize $pq$};
      
      \node [mundo] (w001) at (0, -2) {\mbox{}};
      \node [mundo] (w011) at (2, -2) {\scriptsize $p$};      
  
      \node [texto, node distance=10pt, left=of w000.north]
      {\scriptsize $((w_{0},0),0)$}; 
      \node [texto, node distance=10pt, right=of w010.north]
      {\scriptsize $((w_{0},1),0)$}; 
      \node [texto, node distance=10pt, left=of w001.south]
      {\scriptsize $((w_{0},0),1)$}; 
      \node [texto, node distance=10pt, right=of w011.south]
      {\scriptsize $((w_{0},1),1)$}; 
  
      \path (w000) edge [loop left] (w000)
                   edge [<->] (w001)
                   edge [<->] (w010)
                   edge [<->] (w011)
            (w001) edge [loop left] (w001)
                   edge [<->] (w010)
                   edge [<->] (w011)
            (w010) edge [loop right] (w010)
                   edge [<->] (w011)
            (w011) edge [loop right] (w011);
    \end{tikzpicture}
    \\
    $\model^{\set{\set{\lnot p,\lnot q}}}_\su$
    &
    $\left(\model^{\set{\set{\lnot p,\lnot q}}}_\su\right)^{\set{\set{q}}}_\su$
  \end{tabular}
  \end{center}
  Observe that in the resulting model the agent can access the state
  $((w_0,1),1)$ where $p\land \lnot q$ is true, so
  \[\model,w_0 \modelsu \psfor{\pi}{\psfor{\lnot
      \pi}{\ps{(p\land\lnot q)}}}
  \]
  But with forgetting whether $\pi$, starting at $\model$ it is not possible to produce a state where $p\land \lnot q$ is true. With
  independence of the chosen clause in $\D(\lnot \pi)$, we arrive at
  the following model: 
  
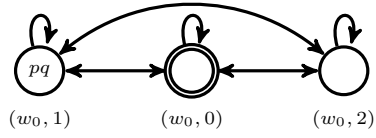
\begin{center}
    \begin{tikzpicture}[->,>=stealth', very thick]
      \node [mundo, double] (w00) at (0, 0) {\mbox{}};
      \node [mundo] (w01) at (-2, 0) {\scriptsize $pq$};
      \node [mundo] (w02) at (2, 0) {\mbox{}};

      \node [texto, node distance=3pt, below=of w00.south]
      {\scriptsize $(w_{0},0)$}; 
      \node [texto, node distance=3pt, below=of w01.south]
      {\scriptsize $(w_{0},1)$}; 
      \node [texto, node distance=3pt, below=of w02.south]
      {\scriptsize $(w_{0},2)$}; 
      
      \path (w00) edge [loop above] (w00)
                  edge [<->] (w01)
                  edge [<->] (w02)
            (w01) edge [loop above] (w01)
                  edge [<->, bend left = 40] (w02)
            (w02) edge [loop above] (w02);
    \end{tikzpicture}
  \end{center}
  Then, $\model,w_0 \not\modelsu \psforabt{\pi}{\ps{(p\land\lnot q)}}.$
  \end{proof}
\end{proposition}

{\medskip}

Another difference between forgetting and forgetting whether is that, while it is not possible to {\em forget the truth-value} of a contradiction, it {\em is} possible to {\em forget} a contradiction, as  $\D(\bot) \neq \varnothing$. In fact, $\models \ncfor{\bot}{\varphi}\ldimp\varphi$ and $\models\ncfor{\bot}{\varphi}\ldimp \psfor{\bot}{\varphi}$. Note, however, that if the agent knows a contradiction, the action of forgetting (the contradiction itself or any other formula) will not `fix' this. In fact, the action cannot turn an agent's knowledge contradictory or consistent if it was not that way before. 

\begin{proposition}
  Let $\pi$ be any propositional formula that is neither a tautology 
  nor a contradiction. Then, 
  \[ \modelsu \nc{\bot} \ldimp \ncfor{\pi}{\nc{\bot}}.\]
  \begin{proof}
    For the left-to-right direction, take any pointed model $(\model, w)$. The antecedent $\nc{\bot}$ states that $w$ has no successors and hence, by Definition \ref{def:multiclause}, neither does $(w,0)$ regardless of the chosen clause $D \in \D(\pi)$; thus, $\ncfor{\pi}{\nc{\bot}}$.

    For the other direction, argue by contrapositive. Assume that $\model,w\models\lnot\nc{\bot}$. Then, $w$ has at least one successor and hence, by Definition \ref{def:multiclause}, so does $(w,0)$ regardless of the chosen clause $D \in \D(\pi)$; thus, $\ncfor{\pi}{\lnot\nc{\bot}}$.
  \end{proof}
\end{proposition}

With respect to tautologies, forgetting behaves as forgetting whether: the clausal form of $\top$ is $\varnothing$, and therefore it is not possible to forget a tautology.

{\bigskip}

When compared with the action of forgetting whether, the action of forgetting is closer to the well-known action of belief contraction: both represent an epistemic action after which the agent does not know/believe a given formula, regardless of the epistemic attitude towards the formula's negation. This allows a more accurate comparison with a key concept within belief contraction: that of \emph{recovery}.

{\smallskip}

An action of forgetting a given $\pi$ might have side-effects: the agent might also forget a second formula $\varphi$. In such cases it seems desirable for an action of `remembering' $\pi$ to make the agent to remember $\varphi$ too.\footnote{Still, within \textit{AGM}, the recovery postulate is the most discussed, as there are examples showing that such behaviour is not always reasonable. See, e.g., \cite{DBLP:journals/jolli/Hansson93,DBLP:journals/jolli/Ferme98,Ferme2001}.} The forgetting action of this section satisfies a form of recovery, restricted to cases in which $\pi$ was known to begin with. For describing this we will use the public announcement operation in public announcement logic \cite{Plaza1989,GerbrandyGroeneveld1997}, represented syntactically with formulas of the form $\anuncio\pi$, as it matches closely the semantic nature of this approach.

\begin{proposition}
If $\pi$ is a propositional
  formula and $\varphi$ an arbitrary formula of $\mathcal{L}_{\Box}$ then
  \[ (\nc\pi\land \varphi) \limp  \ncfor{\pi}[\anuncio\pi]\varphi\]
is valid over the class of transitive models.
\end{proposition}
  
\proof
We will assume familiarity with the semantics of $\anuncio\pi$ and give only a sketch of the argument. Suppose that $\model$ is a transitive model and $w$ is a world of $\model$ such that $\model, w\models \nc\pi\land\varphi$. If $D\in \D(\pi)$ is arbitrary, then $\model_u^{\set D}$ is obtained from $\model$ by adding a copy of each world where $D$ fails, but after applying $\anuncio\pi$, we pass to the model $(\model_u^{\set D})_{\anuncio\pi}$ where all such worlds are deleted. Since every world accessible from $w$ (and hence, by transitivity, in the submodel generated by $w$) already satisfied $\pi$, the submodel generated by $w$ in $\model$ is isomorphic to the submodel generated by $(w,0)$ in $(\model_u^{\set D})_{\anuncio\pi}$. Hence $(\model_u^{\set D})_{\anuncio\pi},(w,0)\models \varphi$. It follows that $\model_u^{\set D},(w,0)\models [\anuncio\pi ]\varphi$ and, since $D$ was arbitrary, that $\model ,(w,0)\models \ncfor{\pi} [\anuncio\pi ]\varphi$. Since $\model$ and $w$ were also arbitrary, the claim follows.
\endproof

\section{An axiomatization}\label{SecAx}

The operation of Definition \ref{def:multiclause}, with a finite (possibly empty) set of non-tau\-to\-log\-i\-cal clauses $\D$ as a parameter, produces a new model that contains $|\D|+1$ copies of the original one, with one copy being identical to the original and with the rest being the result of falsifying uniformly each one of the clauses in $\D$. This and other similar effects can be achieved by so-called \emph{action models}.

\begin{definition}[Action model]
Let $\mathcal L$ be a formal language that can be interpreted over the models of Definition \ref{def:model}. An \emph{$\mathcal L$-action model} is defined as a tuple $\U = \tupla{\E, \R, \pre, \post}$, where $\E$ is a non-empty set of actions, $\R \subseteq \E \times \E$ is a binary relation, $\pre: \E \to \mathcal L$ a {\em precondition function} assigning a formula of $\mathcal L$ to each action in $\E$, and $\post: (\E \times  \at )  \to \mathcal L  $ a {\em postcondition function} assigning a formula of $\mathcal L$ to each pair of atom in $\at$ and action in $\E$. A \emph{pointed action model} is a pair $(\U, \e)$ where $\U$ is an action model and $\e$ an element of its domain.
\end{definition}

Action models are intended to be applied to relational models; such application produces a new relational model, defined as follows.

\begin{definition}[Product update]\label{def:productpdate}
  Let $\model = \langle W, R, V \rangle$ be a model and $\U = \tupla{\E, \R, \pre, \post}$ an action model. The new model $\model \otimes \U = \tupla{W', R', V'}$ is given by:
\begin{itemize}
\item $W' := \set{(w, \e) \in (W \times \E) \mid \model, w \models \pre(\e)}$;
\item $(w,\e) R' (u, \f) \;\;\text{iff}\;\; wRu \;\text{and}\; \e\R\f$; and
\item for every $p \in \at$, $V'(p) := \set{(w,\e) \in W' \mid \model, w \models \post(\e,p)}.$
\end{itemize}
\end{definition}

In words, the new model's domain is the restricted Cartesian product between $\model$'s and $\U$'s: $(w, \e)$ is a world in $\model \otimes \U$ if and only if $w$ satisfies $\e$'s precondition. In this new model, the agent cannot distinguish world $(u, \f)$ from world $(w, \e)$ if and only if she did not distinguish $u$ from $w$ in $\model$ and could not distinguish $\f$ from $\e$ in $\U$. Finally, a  world $(w, \e)$ satisfies an atom $p$ if and only if $w$ satisfied $p$'s postcondition at $\e$ in $\model$.

Action models will be useful to us since the model operation of Definition \ref{def:multiclause} can be represented by the specific action model described below.
\begin{definition}\label{def:ouractionmodel}
  Let $\D=\{ D_i : i\in I\}$ be a finite (possibly empty) set of
  non-tautological clauses. The action model $\U_{\D} = \tupla{\E,
    \R, \pre, \post}$ is given by 
  \[
    \E := \set{ \e_{i}}_{i\in \set 0\cup I},
    \qquad
    \R := \E \times \E,
    \qquad
    \pre(\e_{i}) := \top\text{ for all }i\in \set 0\cup I,
  \]
  for every $p \in \at$, $\post(\e_{0},p) := p$ and, for $i \in I$,
\[    \;
    \post(\e_{i},p) := \left\{
      \begin{array}{ll@{}}
        p    & \text{if } \set{p, \lnot p} \cap D_{i} = \varnothing; \\
        \top & \text{if } \lnot p \in D_{i}; \\
        \bot & \text{if } p \in D_{i} .
      \end{array}
    \right.
  \]
\end{definition}

  \begin{example}
    Consider $\pi = p\land q$, and recall that
    $\D(\pi)=\set{\set{p},\set{q}}$ and 
    $\D(\lnot \pi) = \set{\set{\lnot p,\lnot q}}$. Then, the action
    model $\U_{\set{\set{p},\set{\lnot p,\lnot q}}}$, defined using one
    clause in $\D(\pi)$ and one in $\D(\lnot\pi)$, is given by:
      
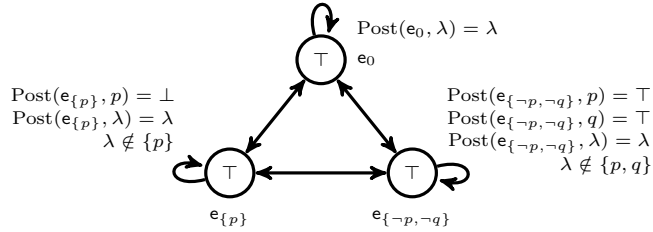
\begin{center}
    \begin{tikzpicture}[->,>=stealth', very thick]
      \node [mundo] (e0) at (0, 1.5) {\scriptsize $\top$};
      \node [mundo] (e1) at (-1.2, 0) {\scriptsize $\top$};
      \node [mundo] (e2) at (1.2, 0) {\scriptsize $\top$};

      \node [texto, node distance=2pt, right=of e0.east]
      {\scriptsize $\e_0$}; 
      \node [texto, node distance=2pt, below=of e1.south]
      {\scriptsize $\e_{\set{p}}$}; 
      \node [texto, node distance=2pt, below=of e2.south]
      {\scriptsize $\e_{\set{\lnot p,\lnot q}}$}; 

      \node [texto, align=left] at (1.4,1.9)
      {\scriptsize $\post(\e_{0},\lambda) = \lambda$};

      \node [texto, align=left] at (-3,1)
      {\scriptsize $\post(\e_{\set{p}},p) = \bot$}; 
      \node [texto, align=left] at (-3,.7)
      {\scriptsize $\post(\e_{\set{p}},\lambda) = \lambda$}; 
      \node [texto, align=left] at (-2.4,.4)
      {\scriptsize $\lambda\notin \set{p}$};

      \node [texto, align=left] at (3,1)
      {\scriptsize $\post(\e_{\set{\lnot p,\lnot q}},p) = \top$}; 
      \node [texto, align=left] at (3,.7)
      {\scriptsize $\post(\e_{\set{\lnot p,\lnot q}},q) = \top$}; 
      \node [texto, align=left] at (3,.4)
      {\scriptsize $\post(\e_{\set{\lnot p,\lnot q}},\lambda) = \lambda$}; 
      \node [texto, align=left] at (3.75,.1)
      {\scriptsize $\lambda\notin\set{p,q}$};

      \path (e0) edge [loop above] (e0)
                  edge [<->] (e1)
                  edge [<->] (e2)
            (e1) edge [loop left] (e1)
                  edge [<->] (e2)
            (e2) edge [loop right] (e2);
    \end{tikzpicture}
  \end{center}
  Preconditions are represented inside the states. For states other than $\e_0$, the postconditions are set up to falsify each state's respective clause.
  \end{example}

Note how, in every action model $\U_{\D}$, the relation $\E$ is the full Cartesian product. Thus, the upgrade operation of Definition \ref{def:productpdate} preserves many relational properties, including seriality, reflexivity, symmetry, transitivity and euclideanity.
These action models give us an alternative representation of the
models $\model^{\D}_u$:

\begin{proposition}
  Let $\model$ be a model and $\D$ a finite (possibly empty) set of non-tautological clauses. Then, the models
  $\model^{\D}_u$ from Definition \ref{def:multiclause}
  and $\model \otimes \U_{\D}$ from Definitions
  \ref{def:ouractionmodel} and \ref{def:productpdate} are isomorphic. 
  \begin{proof}
    Take a model $\model = \langle W, R, V \rangle$; it will be proved that 
    $\model^{\D}_u = \langle W^\D_u,
    R^\D_u, V^\D_u\rangle$ and
    $\model \otimes \U_{\D} = \langle W', R', V'
    \rangle$ are isomorphic, witness the function $f\colon W^\D_u\to W'$ given by $f(w,i)=(w,\e_i)$. 

    First, note how $(w,i)R^\D_u (v,j)$ iff $(w,\e_i) R' (v,\e_j)$. This is because, by Definition~\ref{def:multiclause},
    $(w,i) R^\D_u (v,j)$ iff $wRv$. Moreover, by Definition~\ref{def:productpdate}, $(w,\e_i) R' (v,\e_j)$ iff $wRv$ and
    $\e_i\R\e_j$. But, by Definition~\ref{def:ouractionmodel}, $\R$ is the
    total relation in $\E$, so $(w,i) R^\D_u (v,j)$ iff $(w,\e_i) R'
    (v,\e_j)$. 

    Now, to prove that $(w,i)\in V^{\D}_\su(p)$ iff
    $f(w,i)\in V'(p)$, observe that, by Definition~\ref{def:multiclause},
    $(w,0)\in V^{\D}_\su(p)$ iff $w\in V(p)$. By
    Definition~\ref{def:productpdate}, $(w,\e_0)\in V'(p)$ iff
    $\model,w \models \post(\e_0,p)$, and by
    Definition~\ref{def:ouractionmodel}, $\post(\e_0,p) = p$, so
    $(w,\e_0)\in V'(p)$ iff $\model,w\models p$ iff $w\in V(p)$ iff $(w,0)\in V^{\D}_\su(p)$. For $i\neq 0$, by
    Definition~\ref{def:multiclause}, $(w,i)\in V^{\D}_\su(p)$ iff
    \begin{equation}
      \label{eq:2}
      \mbox{either}\quad \set{p,\lnot p}\cap D_i = \varnothing \quad
      \mbox{and} \quad w\in V(p), \quad \mbox{or else}\quad \lnot
      p\in D_i.
    \end{equation}
    By Definition~\ref{def:productpdate},
    \begin{equation}
      \label{eq:1}
      (w,\e_i)\in V'(p) \quad\mbox{iff}\quad \model,w\models
      \post(\e_i,p);
    \end{equation}
    but by Definition~\ref{def:ouractionmodel}, $(\model,w)$ may only
    satisfy $\post(\e_i,p)$ in two cases, when it is $p$ and $\top$,
    as $\bot$ is unsatisfiable. Then, recalling~\eqref{eq:1},
    $(w,\e_i)\in V'(p)$ iff 
    \begin{equation}
      \label{eq:3}
      \model,w\models p\quad\mbox{and}\quad \set{p,\lnot p}\cap D_i
      = \varnothing,\quad\mbox{or else}\quad
      \model,w\models\top\quad\mbox{and}\quad \lnot p\in D_i. 
    \end{equation}
    Note the equivalence of~\eqref{eq:2} and~\eqref{eq:3}, which
    proves that $(w,i)\in V^{\D}_\su(p)$ iff $(w,\e_i)\in V'(p)$. 
  \end{proof}
\end{proposition}

This correspondence of our model operation with the effect of $\U_{\D}$ allows the use of the action models machinery for obtaining an axiom system for the language $\mathcal{L}_{\Box\ModForAbt}$ with respect to our semantic models. First, recall the definition of the satisfaction relation for action model logic.

\begin{definition}
 Let $\AMA$ be a set of \emph{finite} pointed $\LanBox$-action models, that is, a set containing pointed $\LanBox$-action models (thus with precondition and postconditions functions returning formulas in $\LanBox$) whose domain is finite (non-empty) and in which each action affects the truth-value of at most a finite number of atomic propositions.\footnote{This finiteness condition is required to allow each pointed action model $\ncame{}$ to be associated to a syntactic object and thus to be used as a modality within formulas. For more details, the reader is referred to Section 6.1 of \cite{vanDitmarschEtAl2007}.} The language ${\mathcal L}_{\nc\AMA}$ extends $\LanBox$ with new formulas of the form $\ncame\varphi$ with $(\U,\e)\in\AMA$. Let $\model = \tupla{W, R,V}$ and $w\in W$. The satisfaction relation of Definition \ref{DefSatBasic} is extended by setting $\model, w \models \ncame{\varphi}$ if and only if
\[\model, w \models \pre(\e) \quad\Rightarrow\quad M \otimes \U, (w,\e) \models \varphi.\]
  The set of valid formulas of ${\mathcal L}_{\nc\AMA}$ is denoted by ${\rm Log}_{\nc\AMA}$.
\end{definition}

The following result is based on the action model axiomatization provided in Section 6.6 of \cite{vanDitmarschEtAl2007} together with the remarks of \cite{WangCao2013} (the latter in the context of public announcements). Recall that the logic $K$ contains propositional tautologies, modus ponens, the distribution axiom $\vdash \nc(\varphi \limp \psi) \limp (\nc{\varphi} \limp \nc{\psi})$ and the necessitation rule that derives $\vdash \nc{\varphi}$ from $\vdash \varphi$.

\begin{theorem}\label{TheoAxAction} 
  Let $\Lambda$ be any of the logics $K$, $T$, $K4$, $K5$, $S4$, $S5$ and let $\AMA$ be a set of $\Lambda$-complying \emph{finite} pointed $\LanBox$-action models. The logic ${\rm Log}_{\nc\AMA}$ is axiomatized by the modal logic $\Lambda$ together with the following axioms and rules for all $(\U,\e)\in\AMA$:
  \begin{center}
    \begin{small}
      \renewcommand{\arraystretch}{1.2}
      \begin{tabular}{ll}
        \toprule
        $\vdash \ncame{p} \;\ldimp\; \left(\pre(\e) \limp \post(\e,p) \right)$ &
        $\vdash \ncame{\nc{\varphi}} \;\ldimp\;\left(\pre(\e) \limp \bigwedge_{\f \in \R[\e]} \nc{\ncamf{\varphi}} \right)$ \\
        $\vdash \ncame{\lnot \varphi} \;\ldimp\; \left(\pre(\e) \limp \lnot \ncame{\varphi} \right)$ &
        $\vdash \ncame{(\varphi \limp \psi)} \;\limp\; \left(\ncame{\varphi} \limp \ncame{\psi} \right)$ \\
        $\vdash \ncame{(\varphi \land \psi)} \;\ldimp\; \left(\ncame{\varphi} \land \ncame{\psi} \right)$ &
        From $\vdash \varphi$ infer $\vdash \ncame{\varphi}$ \\
        \bottomrule
      \end{tabular}
    \end{small}
  \end{center}
\end{theorem}

This result can be used to obtain an axiom system for our particular modality $\ncforabt{\pi}{}$. Given clauses $D_{1} \in \D(\pi)$ and $D_{2} \in \D(\lnot \pi)$, our axiom system will use the auxiliary modalities $\ncfaln{D_{1}}{D_{2}}{}$, $\ncfalt{D_{1}}{D_{2}}{}$ and $\ncfalb{D_{1}}{D_{2}}{}$, whose semantic interpretation is as follows:
  \begin{flushleft}
    \begin{tabular}{l@{\quad\text{iff}\quad}l}
      $\model, w \models \ncfaln{D_{1}}{D_{2}}{\varphi}$ & $\model^{\set{D_1,D_2}}, (w, 0) \models \varphi$;\\
      $\model, w \models \ncfalt{D_{1}}{D_{2}}{\varphi}$ & $\model^{\set{D_1,D_2}}, (w, 1) \models \varphi$; and\\
      $\model, w \models \ncfalb{D_{1}}{D_{2}}{\varphi}$ & $\model^{\set{D_1,D_2}}, (w, 2) \models \varphi$.\\
    \end{tabular}
  \end{flushleft}

Observe how $\ncfaln{D_1}{D_2}{}$, $\ncfalt{D_1}{D_2}{}$ and
$\ncfalb{D_1}{D_2}{}$ correspond, respectively, to
$\ncam{\Udd}{\e_{0}}{}$, $\ncam{\Udd}{\e_{1}}{}$ and
$\ncam{\Udd}{\e_{2}}{}$, with $\U_{\D}$ the $\LanBox$-action model of
Definition \ref{def:ouractionmodel}. Moreover, note how the use of
$\U_{\D}$ within a modality is proper, as it is a \emph{finite} action
model: it has a non-empty finite domain and, given that both $D_1$ and
$D_2$ are finite (Lemma \ref{lem:clauses}), each one of its actions
changes the truth-value of at most a finite number of atomic
propositions. Finally, note how $\ncforabt{\pi}{\varphi}$ is
equivalent to $\bigwedge_{D_1\in\D(\pi)}\bigwedge_{D_2\in\D(\neg
  \pi)}\ncfaln{D_1}{D_2}{\varphi}$. Our axiomatization is, then, as
follows:

\begin{definition}[The axiom system ${\rm Ax}_{\nc\ModForAbt\AMD}$]
Let $\AMD$ be the set of all pointed action models of the form $({\Udd},{\e_{i}}){}$ with $i\in \{0,1,2\}$ (recall that $\pre(\e_{0}) = \pre(\e_{1}) = \pre(\e_{2})$) and $D_j$ a finite non-tautological clause; let $\mathcal L_{\nc\ModForAbt\AMD}$ be the extension of $\LanFor$ with expressions of the form $[{\U},\e]\varphi$ with $({\U},\e)\in\AMD$.

The set of axioms ${\rm Ax}_{\nc\ModForAbt\AMD}$ is defined by the following schemas:
\begin{center}
  \begin{small}
    \renewcommand{\arraystretch}{1.2}
    \begin{tabular}{ll}
      \toprule
      $\vdash\ncforabt{\pi}{\varphi}
    \;\ldimp\;
    \bigwedge_{D_1 \in \D(\pi)}\bigwedge_{D_2\in\D(\neg \pi)} \ncfaln{D_1}{D_2}{\varphi}$&\\
      \midrule
      $\vdash \ncfaln{D_1}{D_2}{p} \;\ldimp\; p$    & for all $p \in \at$ \\
      \midrule
      $\vdash \ncfalt{D_1}{D_2}{p} \;\ldimp\; p$    & if $\set{p, \lnot p} \cap D_{1} = \varnothing$ \\
      $\vdash \ncfalt{D_1}{D_2}{p} \;\ldimp\; \top$ & if $\lnot p \in D_{1}$ \\
      $\vdash \ncfalt{D_1}{D_2}{p} \;\ldimp\; \bot$ & if $p \in D_{1}$ \\
      \midrule
      $\vdash \ncfalb{D_1}{D_2}{p} \;\ldimp\; p$    & if $\set{p, \lnot p} \cap D_{2} = \varnothing$ \\
      $\vdash \ncfalb{D_1}{D_2}{p} \;\ldimp\; \top$ & if $\lnot p \in D_{2}$ \\
      $\vdash \ncfalb{D_1}{D_2}{p} \;\ldimp\; \bot$ & if $p \in D_{2}$ \\
      \midrule      
      $\vdash \ncfaldd{\lnot \varphi} \;\ldimp\; \lnot \ncfaldd{\varphi}$ \\
      $\vdash \ncfaldd{(\varphi \land \psi)} \;\ldimp\; \left(\ncfaldd{\varphi} \land \ncfaldd{\psi} \right)$ \\
      $\vdash \ncfaldd{\nc{\varphi}} \;\ldimp\; \nc{\left(\ncfaln{D_1}{D_2}{\varphi} \land \ncfalt{D_1}{D_2}{\varphi} \land \ncfalb{D_1}{D_2}{\varphi}\right)} $ \\
      $\vdash \ncfaldd{(\varphi \limp \psi)} \;\limp\; \left(\ncfaldd{\varphi} \limp \ncfaldd{\psi} \right)$ \\
      From $\vdash \varphi$ infer $\vdash \ncfaldd{\varphi}$ \\
      \bottomrule
    \end{tabular}
  \end{small}
\end{center}
where $\ncfaldd{}$ is any of $\ncfaln{D_1}{D_2}{}$, $\ncfalt{D_1}{D_2}{}$ and $\ncfalb{D_1}{D_2}{}$ and $\pi$ is propositional.
\end{definition}

It is now possible to state this paper's main result:

\begin{theorem}\label{TheoFAxiom}
Let $\Lambda$ be any of the logics $K$, $T$, $K4$, $K5$, $S4$, $S5$. A formula $\varphi\in \LanForAct$ is valid over the class of $\Lambda$-complying models if and only if $\Lambda+{\rm Ax}_{\nc\ModForAbt\AMD}\vdash\varphi$.
\end{theorem}

\proof Soundness is immediate since all axioms are true and all rules preserve validity, where we appeal to Theorem \ref{TheoAxAction} for those axioms involving action models.

For completeness, let $\varphi\in \LanForAct$ be a valid formula. Then, by the first axiom of ${\rm Ax}_{\nc\ModForAbt\AMD}$, the formula $\varphi$ can be replaced by a provably equivalent formula $\widehat\varphi\in {\mathcal L}_{\nc\AMD}$. By Theorem \ref{TheoAxAction}, $\widehat \varphi$ is derivable, hence so is $\varphi$.
\endproof

Since the action models in $\AMD$ were auxiliary, it might be useful to restate this result in terms of our original language:

\begin{corollary}
Let $\Lambda$ be any of the logics $K$, $T$, $K4$, $K5$, $S4$, $S5$. A formula $\varphi\in \LanFor$ is valid over the class of $\Lambda$-complying models if and only if it is derivable in $\Lambda+{\rm Ax}_{\nc\ModForAbt\AMD}$. 
\end{corollary}

\section{Alternative forgetting operators}\label{SecAlt}

There are many possibilities when modelling the action of forgetting. Our aim has been to present a semantic approach, rather than a syntactic one, but even then there are several routes one can take. In this section we mention some variations and discuss how they relate to our proposal.

\subsection{Conditionally forgetting}

The defined actions of {\em forgetting $\pi$} (Definition \ref{def:for}) and {\em forgetting whether $\pi$} (Definition \ref{def:forabt}) do not have any precondition, and thus they can take place regardless of whether the agent knows $\pi$ (for forgetting $\pi$) or knows either $\pi$ or else $\lnot \pi$ (for forgetting whether $\pi$). This choice has been made because, technically, there is no reason to restrict the respective operations: they can be applied to any model, regardless of the agent's epistemic attitude towards $\pi$.\footnote{Compare this with the precondition for public announcements. In order to be announced, a formula needs to be true, not only because of the interpretation of the operation (public \emph{and truthful} announcements), but also for technical reasons: if the formula is false, then the evaluation point will be removed, and thus it is not possible to evaluate formulas in it after the operation.} As a result, even though anomalies may occur in `abnormal' situations, the defined actions work as expected in the intended cases (the agent knows / knows whether the --non tautological and non contradictory-- formula she is forgetting is the case).

However, it is also interesting to assume that the agent would not act unless she is in an intended situation. Let us focus on the action of forgetting $\pi$. An interesting possibility is working as in Definition  \ref{def:for} when the agent knows $\pi$ but doing nothing otherwise.

  \begin{definition}\label{DefCondFor}
For a propositional formula $\pi$, define a modal operator $\ncfor{'\pi}$ and extend the semantics in Definition \ref{DefSatBasic} by setting
  \begin{equation*}
    \model, w \modelsu \ncfor{'\pi}\varphi
    \quad\text{iff}\quad
    \begin{cases}
      \forall D\in \D(\pi), \model^{\set{D}}_\su, (w,0)
      \modelsu\varphi & \text{if } \model, w \modelsu \nc\pi\\
      \model, w \models \varphi & \text{otherwise.}
    \end{cases}    
  \end{equation*}
Thus, $\varphi$ should be always evaluated, and the precondition only determines where: in $\model^{\set{D}}_\su$ for every $D\in \D(\pi)$ when the precondition holds, and only in $\model$ otherwise.\footnote{Note how this is different from other alternatives, as the one used in public announcement logic when evaluating $[\anuncio{\pi}]\varphi$: if the announced formula $\pi$ is not true, then $\varphi$ does not need to be evaluated, and $[\anuncio{\pi}]\varphi$ holds by vacuity.}
  \end{definition}

This new operator has several properties that are interesting when compared to other approaches as belief contraction. For example, a vacuity principle is immediate:

\begin{proposition}
If $\pi$ is propositional and $\varphi$ is an arbitrary formula then
\begin{enumerate}

\item $\lnot\nc\pi \limp (\varphi \ldimp \ncfor{'\pi}\varphi)$ is valid, but

\item $\lnot\nc\pi \limp (\varphi \ldimp \ncfor{\pi}\varphi)$ is not necessarily valid.
 
 \end{enumerate}
  \end{proposition}

\proof
The first claim is immediate from Definition \ref{DefCondFor}.

As a counterexample for the second claim we may take $\pi=p\wedge q$ and $\varphi=\nc p$, then consider a model with a single reflexive world $w$ satisfying $p\wedge\neg q$. Then, $\model,w\models \nc p$, but $\model_u^{\set{\set p}},(w,0)\models \neg\nc p$, and hence
\[\model, w\models \neg \nc \pi \wedge \neg(\varphi \leftrightarrow \ncfor \pi \varphi).\qedhere\]
\endproof

Thus the two notions of forgetting behave differently. However, in order to study them together, it is not necessarily to have both as primitives: the version with precondition can be defined in terms of its more general counterpart.
  \begin{proposition}
  If $\pi$ is propositional and $\varphi$ is arbitrary, then 
  \[ \models \ncfor{'\pi}\varphi \ldimp (\neg\nc\pi\wedge\varphi)\vee(\nc\pi\wedge \ncfor\pi\varphi).\]
  \end{proposition}

A similar variant could be defined of $\ncforabt\pi$. We will not go into details, but the treatment here would be more nuanced, as we would have to consider three cases, depending on whether $\nc\pi$, $\nc\neg\pi$, or neither one of them holds.

\subsection{Strongly forgetting whether}

As another natural alternative, a deterministic version of our operator can be
explored which would entail that the agent loses {\em all} information regarding $\pi$. According to Definition \ref{def:forabt}, in order to check
$\ncforabt{\pi}{\varphi}$, we need to check $\varphi$ in several
models, one for each element in $\D(\pi)\times\D(\lnot\pi)$. The
number of models to check can be exponential. An 
alternative deterministic forgetting operator could create just one
model by appending a new copy of each world for each clause in
$\D(\pi)\cup\D(\lnot\pi)$.

\begin{definition}
For a propositional formula $\pi$, define a modal operator $\ncforabt{^\ast\pi}$ and extend the semantics in Definition \ref{DefSatBasic} by setting
  \begin{equation}
    \model, w \modelsu \ncforabt{^{\ast}\pi}\varphi
    \quad\text{iff}\quad \model_u^{\D(\pi)\cup \D(\neg \pi)}\models \varphi.
  \end{equation}
\end{definition}

So now we need only check one model, but the price to pay is that the new model may be exponentially larger than the original. As before, the alternative operator
$\ncforabt{^{\ast}\pi}{\varphi}$ would have different properties to
$\ncforabt{\pi}{\varphi}$. For example, we have the following:

\begin{proposition}
Over the class of serial models,
\begin{enumerate}

\item $\ncforabt{^\ast(p\land q)}{\left(
  \lnot\nc p \land \lnot\nc q\right)}$ is valid, but

\item $\ncforabt{(p\land q)}{\left(
  \lnot\nc p \land \lnot\nc q\right)}$ is not valid.
\end{enumerate}

\end{proposition}

\proof
Suppose that $\mathcal M$ is a serial model, $w$ is a world of $\model$ and $v$ a world accessible from $w$. The clausal form of $p\wedge q$ is $\set{\set p,\set q}$, while the clausal form of $\neg(p\wedge q)$ is $\set{\set {\neg p,\neg q}}$. Thus $\D(p\wedge q)\cup \D(\neg(p\wedge q))=\{\set p,\set q,\set{\neg p,\neg q}\}$ and in $\model_u^{\D(p\wedge q)\cup \D(\neg(p\wedge q))}$ from $(w,0)$ is accessible a world $(v,\set p)$ satisfying $\neg p$ and another $(v,\set q)$ satisfying $\neg q$. It follows that
\[\model,w\models\ncforabt{^\ast(p\wedge q)}(\neg\nc p\wedge \neg \nc q).\]
Since $\mathcal M$ and $w$ were arbitrary, the first claim follows.

For the second, consider a model consisting of a single reflexive point $w$ satisfying $p\wedge q$. Then, $(\model_u^{\set{\set{q},\set{\neg p,\neg q}}},w)$ clearly satisfies $\nc p$, so $\model,w\not \models \ncforabt{(p\land q)}{\left(
  \lnot\nc p \land \lnot\nc q\right)}$, and hence this formula is not valid.
\endproof

\subsection{Dependent forgetting}\label{SecDepend}

In uniform forgetting (Definition \ref{def:forabt}), an agent forgets a formula $\pi$ by falsifying a fixed clause of $\pi$'s clausal form in one copy of the initial model and a fixed clause of $\lnot\pi$'s clausal form in another. However, it may be that at every point in the model, she forgets $\pi$ `for a different reason'. This gives an alternative way to model forgetting, which behaves in a different way from uniform forgetting as presented above. Let us give the definitions.

\begin{definition}
  Let $\model = \langle W, R, V \rangle$ be a model and $(\D_1,\D_2)$ be two sets of clauses. A {\em forgetting function pair} is a pair of functions $(f_1,f_2)$ such that $f_1\colon W\to \D_1$ and $f_2\colon W\to \D_2$.

  The model $\model^{(f_1,f_2)}_d = \langle W^{(f_1,f_2)}_d, R^{(f_1,f_2)}_d, V^{(f_1,f_2)}_d \rangle$ is given by
  \[
    W^{(f_1,f_2)}_d = W \times \set{0, 1, 2},
    \qquad
    (w,i)R^{(f_1,f_2)}_d (v,j) \;\;\text{iff}\;\; wRv
  \]
  and, for each $p \in \at$:
  \begin{enumerate}
    \item $(w,0) \in V^{(f_1,f_2)}_d(p)$ iff $w \in V(p)$; and
    \item for $i = 1, 2$, $(w,i) \in V^{(f_1,f_2)}_d(p)$ iff either $\set{p, \lnot p} \cap f_i(w) = \varnothing$ and $w \in V(p)$, or else $\neg p \in f_i(w)$.
  \end{enumerate}
\end{definition}

Thus, the model $\model^{(f_1,f_2)}_d$ contains three copies of $\model$. The elements of the first, the $(w,0)$ worlds, have the original atomic valuation; each element of the second, the $(w, 1)$ worlds, falsifies a particular clause in $\D_1$, as indicated by the forgetting function $f_1$; finally, each element of the third copy, the $(w, 2)$ worlds, falsifies a particular clause in $\D_2$, as indicated by the forgetting function $f_2$.

\begin{definition}
Let $\model$ be a model. Define the {\em dependent satisfaction relation} $\modelsd$ on $\model$ by extending Definition \ref{DefSatBasic} with $\model, w \modelsd \ncforabt\pi\varphi$ if and only if 
\[\forall f_1\colon W\to \D(\pi),\forall f_2\colon W\to\D(\neg \pi), \model^{(f_1,f_2)}_d, (w,0) \modelsd\varphi.
  \]
The set of formulas in $\LanFor$ valid under $\modelsd$ will be denoted $\LogForD$.
\end{definition}

It turns out that our dependent and uniform interpretations give rise to different logics. 

\begin{proposition} \label{prop:diffLogics}
The logics $\LogFor$ (for uniform forgetting) and $\LogForD$ (for dependent forgetting) are different. In particular, if
\[\varphi=\nc (p\wedge q)\rightarrow\ncforabt{(p\wedge q)}(\nc p\vee\nc q),\]
then $\varphi\in\LogFor\setminus\LogForD$; this is true even if we restrict to the class of $\sf S5$ models.
\end{proposition}

\proof
We argue semantically that $\varphi\in\LogFor$. Let $\model=\langle W,R,V\rangle$ be a model and $w\in W$ satisfy $\nc(p\wedge q)$.  We have that $\D(p\wedge q)=\set{\set{p},\set{q}}$ and $\D(\neg(p\wedge q))=\set{\set{\neg p,\neg q}}$, so to check that $\ncforabt{(p\wedge q)}(\nc p\vee\nc q)$ holds in $w$ it is enough to see that $(\model^{\set{\set{p},\set{\neg p,\neg q}}}_\su,(w,0))$ and $(\model^{\set{\set{q},\set{\neg p,\neg q}}}_\su,(w,0))$ both satisfy $\nc p\vee \nc q$.

Let $\model^{\set{\set{p},\set{\neg p,\neg q}}}_\su=\langle
W',R',V'\rangle$ and suppose that $(w,0) R' (v,i)$; we claim that
independently of $i$, $(v,i)\in V'(q)$. If $i=0$, this follows from
the assumption that $w\in\ts{p\wedge q}^{\model}$. If $i=1$, then $q$ does not
occur in $\set{p}$ and then $(v,1)\in V'(q)$. Finally, if $i=2$, then
$\neg q\in \set{\neg p,\neg q}$ so $(v,2)\in V'(q)$. In all three
cases, $(v,i)\in V'(q)$ and since $(v,i)$ was arbitrary, $(w,0)\in
\ts{\nc q}^{\model^{\set{\set{p},\set{\neg p,\neg q}}}_\su}$. 

A symmetric argument shows that $\model^{\set{\set{q},\set{\neg p,\neg q}}}_\su,(w,0)\modelsd \nc p$, so that $\model,w\models \ncforabt{(p\wedge q)}(\nc p\vee \nc q),$ as claimed.

{\medskip}

It remains to check that $\varphi\not\in\LogForD$. For this, consider the model $\model=\langle W,R,V\rangle$ shown below, where $W=\set{w,v}$, $R$ is the full relation on $W$ and $V(p)=V(q)= W$; observe that $\model$ is an $\sf S5$ model.
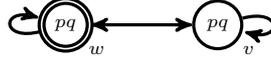
\begin{center}
  \begin{tikzpicture}[->,>=stealth', very thick]
    \node [mundo, double] (w0) at (-1, 0) {\scriptsize $pq$};
    \node [mundo] (w1) at ( 1, 0) {\scriptsize $pq$};

    \node [texto, node distance=6pt, right=of w0.south] {\scriptsize $w$};
    \node [texto, node distance=6pt, right=of w1.south] {\scriptsize $v$};

    \path (w0) edge [loop left] (w0)
               edge [<->] (w1)
          (w1) edge [loop right] (w1);
  \end{tikzpicture}
\end{center}
Clearly, $\model,w\modelsd \nc (p\wedge q)$. To show that $\model,w\not\modelsd \ncforabt{(p\wedge q)}(\nc p\vee \nc q)$, we only need to exhibit two forgetting functions $f_1\colon W\to \set{\set{p},\set{q}}$ and $f_2\colon W\to \set{\set{\neg p,\neg q}}$ such that $\model^{(f_1,f_2)},w\not\modelsd \nc p\vee \nc q$. Let $f_1(w)=\set{p}$ and $f_1(v)=\set{q}$, while $f_2(w) = f_2(v) =\set{\neg p,\neg q}$. The resulting model $\model^{(f_1,f_2)}=\langle W'',R'',V''\rangle$ is the following.
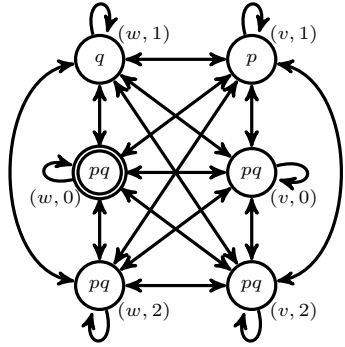
\begin{center}
  \begin{tikzpicture}[->,>=stealth', very thick]
    \node [mundo] (w10) at (-1, 1.5) {\scriptsize $q$};
    \node [mundo] (w11) at ( 1, 1.5) {\scriptsize $p$};
    \node [mundo, double] (w00) at (-1, 0) {\scriptsize $pq$};
    \node [mundo] (w01) at ( 1, 0) {\scriptsize $pq$};
    \node [mundo] (w20) at (-1,-1.5) {\scriptsize $pq$};
    \node [mundo] (w21) at ( 1,-1.5) {\scriptsize $pq$};

    \node [texto, node distance=6pt, right=of w10.north] {\scriptsize $(w,1)$};
    \node [texto, node distance=6pt, right=of w11.north] {\scriptsize $(v,1)$};
    \node [texto, node distance=6pt, left=of w00.south] {\scriptsize $(w,0)$};
    \node [texto, node distance=6pt, right=of w01.south] {\scriptsize $(v,0)$};
    \node [texto, node distance=6pt, right=of w20.south] {\scriptsize $(w,2)$};
    \node [texto, node distance=6pt, right=of w21.south] {\scriptsize $(v,2)$};

    \path (w10) edge [loop above] () 
                edge [<->] (w11)
                edge [<->] (w00)
                edge [<->, bend right = 80] (w20)
                edge [<->] (w01)
                edge [<->] (w21)
          (w00) edge [loop left] ()
                edge [<->] (w01)
                edge [<->] (w11)
                edge [<->] (w21)
          (w20) edge [loop below] ()
                edge [<->] (w21)
                edge [<->] (w00)
                edge [<->] (w01)
                edge [<->] (w11)
          (w11) edge [loop above] ()
                edge [<->] (w01)
                edge [<->, bend left = 80] (w21)
          (w01) edge [loop right] ()
          (w01) edge [<->] (w21)
          (w21) edge [loop below] ();
  \end{tikzpicture}
\end{center}
Observe how $R''$ is also the full relation. Moreover, $(w,1)\not\in
V''(p)$ and $(v,1)\not\in V''(q)$, so $(w,0)\not\in \ts{\nc p\vee\nc
  q}^{\model^{(f_1,f_2)}}$ and hence $\model,w\not\modelsd\ncforabt
{(p \wedge q)}(\nc p\vee\nc q)$, as desired. 
\endproof

Thus, the notion of dependent forgetting leads to a different logic of uniform forgetting. Our example above suggests that, with depending forgetting, more information is lost in the act of forgetting, and this may be desirable in applications. However, the technique of representing forgetting in terms of action models does not work in this setting (at least not in a straightforward way); a further exploration of this notion of forgetting is left for future work.

\section{Conclusions}

The present paper uses the possible world semantics to model the changes that occur in an agent's information when she {\em forgets the truth-value} of a propositional formula as represented by its `minimal' conjunctive normal form. Besides introducing a \emph{uniform forgetting whether} model operation representing such action and its correspondent modality for expressing its effects, the paper has discussed several properties of the operation as well as provided a sound and complete axiom system for it. Two variations of this \emph{uniform forgetting whether} action have been explored: a simpler \emph{uniform forgetting} action which simply forgets that a given formula is the case, and a more complex \emph{dependent forgetting whether} under which the agent might forget the given formula's truth-value for different reasons in different parts of the model. It has been proved not only that \emph{uniform forgetting whether} cannot be defined in terms of the simpler \emph{uniform forgetting}, but also that the \emph{uniform forgetting whether} and the \emph{dependent forgetting whether} give raise to different logics.

{\medskip}

Several directions are left for further study. First, the
axiomatization of the \emph{dependent forgetting} logic is still an
open issue. It is possible that the action models axioms can be used
for \emph{dependent forgetting} if some restrictions are introduced to
the forgetting functions. Second, and possibly more interesting, is an
action of forgetting modal formulas, which would allow the agents to
forget their own or other agents' epistemic states. Finally, our model of forgetting differs in several ways from the belief contraction approach. We leave a more comprehensive comparison of the two, along with a possible unification, for future work.

\end{document}